%% file: root.tex
\documentclass[letterpaper, 10 pt, conference]{ieeeconf}  

\usepackage{color,array}

\usepackage{amssymb,amsfonts}

\usepackage{graphicx}
\usepackage{verbatim}   

\usepackage{algorithm}
\usepackage{algpseudocode}

\usepackage{mathtools}
\usepackage{siunitx}
\usepackage{tabularx}
\usepackage{pgfplots}
\usepackage{pgfplotstable}

\usepackage{multirow}
\usepgfplotslibrary{fillbetween}
\usepackage{booktabs}
\pgfplotsset{compat = newest}
\usetikzlibrary{arrows,positioning,shapes,intersections,external,patterns,calc,fit,decorations,decorations.markings} 
\usepackage{cite}
\input{mydefs.tex}

\IEEEoverridecommandlockouts                              

\title{\LARGE \bf
Physics-informed Learning for Passivity-based Tracking Control
}

\author{Thomas Beckers and Leonardo Colombo
\thanks{Thomas Beckers is with the Department of Computer Science, Vanderbilt University, Nashville, TN 37212, USA {\tt\small thomas.beckers@vanderbilt.edu}}%
\thanks{L. Colombo is with Centre for Automation and Robotics (CSIC-UPM), Ctra. M300 Campo Real, Km 0,200, Arganda
del Rey - 28500 Madrid, Spain.{\tt\small leonardo.colombo@csic.es}}%
}

\begin{document}

\maketitle
\thispagestyle{empty}
\pagestyle{empty}

\begin{abstract}
Passivity-based control ensures system stability by leveraging dissipative properties and is widely applied in electrical and mechanical systems. Port-Hamiltonian systems (PHS), in particular, are well-suited for interconnection and damping assignment passivity-based control (IDA-PBC) due to their structured, energy-centric modeling approach. However, current IDA-PBC faces two key challenges: (i) it requires precise system knowledge, which is often unavailable due to model uncertainties, and (ii) it is typically limited to set-point control. To address these limitations, we propose a data-driven tracking control approach based on a physics-informed model, namely Gaussian process Port-Hamiltonian systems, along with the modified matching equation. By leveraging the Bayesian nature of the model, we establish probabilistic stability and passivity guarantees. A simulation demonstrates the effectiveness of our approach.
\end{abstract}

\section{Introduction}
Passivity-based control is a widely used approach in control theory that ensures stability by leveraging the inherent dissipative properties of a system~\cite{ortega1997passivity}. Rooted in the concept of passivity, this method formulates control laws that prevent energy generation within the system, thus naturally leading to stable and robust behavior. Passivity-based control has been successfully applied in various engineering domains, including robotics~\cite{hatanaka2015passivity}, power systems~\cite{sira1997passivity}, and mechanical structures~\cite{acosta2005interconnection}, due to its ability to handle nonlinearities and external disturbances effectively.

Passivity-based control modifies a system’s energy function and interconnection structure to achieve desired dynamics, making port-Hamiltonian systems (PHS) a natural framework for such control strategies. PHS provide an energy-centric modeling approach that is widely used in control theory, robotics, mechanical engineering, and related fields~\cite{van2000l2}. PHS models represent physical systems as interconnected subsystems, each corresponding to a specific physical domain, such as electronics or mechanics. The Hamiltonian function describes the energy of each subsystem, while power-based port variables define the interconnections, governing energy exchange between components. This structured representation enables a systematic and intuitive analysis of complex physical systems, facilitating the design of efficient and robust control strategies. Due to these advantages, PHS-based control has gained significant research attention in recent years.

In this context, interconnection and Damping Assignment Passivity-Based Control (IDA-PBC) is a widely used control strategy that exploits the energy-based structure of physical systems to achieve stability and desired closed-loop behavior~\cite{ortega2002interconnection}. In this approach, the system’s dynamics are reshaped into a PHS by 
\begin{figure}[t]
\begin{center}
\vspace{0.1cm}
	\includegraphics[width=0.85\columnwidth]{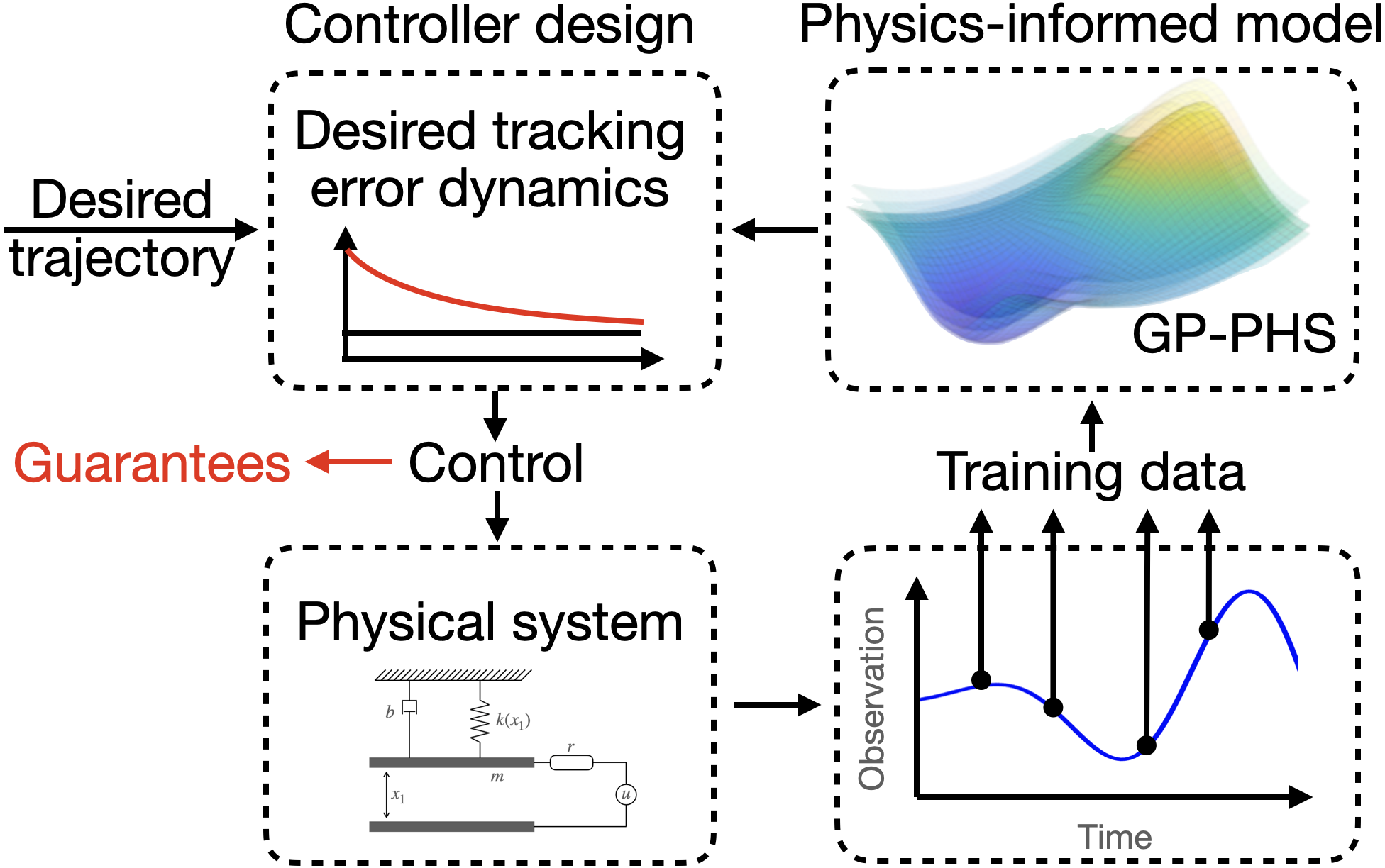}
	\vspace{-0.2cm}\caption{Overview of the GP-PHS based tracking control approach. First, a GP-PHS model is used to learn the partially unknown dynamics of the physical systems. Then, a modified matching equations is used to design a tracking controller that is robustified against the modeling errors, enabling stability and passivity guarantees.}\vspace{-0.7cm}
	\label{fig:gpcphs}
\end{center}
\end{figure}
appropriately modifying the interconnection and damping structure while preserving passivity. The control design consists of two key steps: interconnection assignment, which defines the energy exchange between subsystems to ensure stability, and damping assignment, which introduces dissipative elements to regulate energy flow and enhance robustness. It has been successfully applied to a wide range of systems, including mechanical systems, power systems, and biological systems~\cite{gomez2004physical,ortega2004interconnection,9483212}. Despite its effectiveness, IDA-PBC requires an accurate model of the system’s Hamiltonian function, interconnection matrix, and damping structure, making it challenging to apply in scenarios with significant uncertainties or complex dynamics. Furthermore, traditional IDA-PBC relies on a matching condition, which is primarily suited for shaping potential energy. In contrast, for reference tracking or regulation of many underactuated systems such as the inertial wheel pendulum, see~\cite{wang2008modified}, modifying kinetic energy is also often necessary. Thus, the goal of this paper is to design a passivity-based tracking controller for systems with partially unknown dynamics. 

In \cite{nageshrao2015port}, the authors provide a comprehensive summary of data-driven control approaches for port-Hamiltonian systems. Adaptive control strategies for PHS have been explored in \cite{wang2007simultaneous, dirksz2010adaptive}, while reinforcement learning-based control is presented in \cite{sprangers2014reinforcement}. However, these methods are primarily limited to handling parametric uncertainties and lack rigorous closed-loop stability analysis and passivity guarantees. Alternatively, iterative learning control (ILC) \cite{fujimoto2003iterative} and repetitive control (RC) \cite{fujimoto2004iterative} enable PHS control without requiring a priori system knowledge, but they depend on repetitive executions, which may not be practical in all applications. A robust controller against model uncertainties is introduced in \cite{ryalat2018robust}, yet it suffers from the inherent trade-off between robustness and performance. Additionally, the authors of~\cite{duong2021hamiltonian} propose a Hamiltonian-based Neural ODE approach for tracking control on $SE(3)$ manifolds, but it lacks stability and passivity guarantees, limiting its applicability in safety-critical systems. In \cite{beckers2023data}, we introduced a Gaussian Process Port-Hamiltonian Systems (GP-PHS)-based control approach. However, this method is not generally applicable to tracking control problems.

\textbf{Contribution:} We propose a data-driven passivity-based tracking control approach for physical systems with partially unknown dynamics. First, the unknown dynamics are modeled by a physics-informed learning approach, namely GP-PHS. Based on this model, we use a modified version of the IDA-PBC matching equation to prove the stability of the closed-loop system. Since some of the constraints might be hard to satisfy in practice, we further relax some assumptions but still guarantee probabilistic boundedness of the tracking error and semi-passivity for the closed-loop system. Finally, we apply the proposed method to the tracking control problem of a mechanical system with partially unknown dynamics.

The remainder of the paper is structured as follows. We introduce the idea of IDA-PBC and the problem setting in~\cref{sec:def}, followed by the proposed GP-PHS tracking control law in~\cref{sec:ctrl}. Finally, a simulation shows the benefits of the proposed control approach in~\cref{sec:sim}.

\section{Preliminaries}\label{sec:def}
In this section, we introduce the modeling and control of PHS, followed by the problem setting.
\subsection{Port-Hamiltonian Systems}\label{sec:phs}
Port-Hamiltonian systems provide a powerful framework for modeling and control of physical systems by extending classical Hamiltonian mechanics, see~\cite{van2014port}. They naturally incorporate energy conservation, dissipation, and interconnection structures, making them particularly useful for multi-physics systems such as electrical, mechanical, and fluid systems. In general, a PHS is represented by\footnote{Vectors~$\bm a$ and vector-valued functions~$\bm f(\cdot)$ are denoted with bold characters. Matrices are described with capital letters. $I_n$ is the $n$-dimensional identity matrix and $0_n$ the zero matrix. The expression~$A_{:,i}$ denotes the i-th column of $A$.  $\R_{>0}$ denotes the set of positive real numbers, while $\R_{\geq 0}$ is the set of non-negative real numbers. The operator $\nabla_\x$ with $\x\in\R^n$ denotes $[\frac{\partial}{\partial x_1},\ldots,\frac{\partial}{\partial x_n}]^\top$.}
\begin{align}
\begin{split}
\label{eq:phs}
    \dot{\x} &= \left[J(\x) - R(\x)\right] \nabla_\x H( \x) + G(\x) \u \\
    \bm{y} &= G^\top(\x) \nabla_\x H( \x),
    \end{split}
\end{align}
where $\x \in \mathbb{R}^n$ is the state vector, $H(\x): \mathbb{R}^n \to \mathbb{R}$ is the Hamiltonian function representing the system’s total energy, and $J(\x)=-J(\x)\in\R^{n\times n}$ and $0\preceq R(\x)\in\R^{n\times n}$ define the interconnection and dissipation structures, respectively. The input-output pair $(\u\in\R^m, \bm{y}\in\R^m)$ represents external interactions with the system, a pair of variables whose product gives the (generalized) power that is stored or dissipated by the system.

Using Interconnection and damping assignment passivity-based control (IDA-PBC) \cite{ortega2002interconnection}, it is possible to control a PHS in such a way that it behaves as a target dynamics, namely as a new PHS with a desired interconnection matrix~$J_d$, damping matrix~$R_d$ and energy function~$H_d$. The PHS in~\cref{eq:phs} can be rendered to the desired port-Hamiltonian dynamics described by
\begin{align}
    \dot{\x} = \left[J_d(\x) - R_d(\x)\right] \nabla_\x H_d( \x)
    \label{eq:target_dynamics}
\end{align}
with the control law
\begin{align}
    u = &(G^\top(\x)G(\x))^{-1} G^\top(\x) 
    \Big( [J_d(\x) - R_d(\x)] \nabla_\x H_d( \x)   \notag \\
    &- (J(\x) - R(\x)) \nabla_\x H( \x)\Big),
    \label{eq:control_law}
\end{align}
if and only if the matching equation 
\begin{align}
\begin{split}
    &G^\perp (\x) \Big( [J_d(\x) - R_d(\x)] \nabla_\x H_d( \x) \Big)\\
    =&G^\perp (\x) \Big( [J(\x) - R(\x)] \nabla_\x H( \x) \Big)
    \label{eq:matching_equation}
\end{split}
\end{align}
holds, where $G^\perp (\x)$ is the full rank left annihilator of $G(\x)$. Thus, we can derive a passivity-based control law for the system~\cref{eq:phs} by solving the matching equation~\eqref{eq:matching_equation}. Even though there is no general closed-form solution of~\cref{eq:matching_equation}, there are different strategies to find promising parametrizations of the desired PHS so that the matching equation is satisfied, see~\cite{ortega1997passivity} for more details.
\subsection{Problem Setting}\label{sec:ps}
We consider the problem of designing a tracking control law for a partially unknown physical system whose dynamics can be written in port-Hamiltonian form~\cref{eq:phs}. We assume that we have access to noisy observations $\tilde{\x}(t)\in\R^n$ of the system state $\x(t)\in\R^n$ whose evolution over time $t\in\R_{\geq 0}$ follows~\cref{eq:phs} with $\x(0)\in\R^n$ as the initial state. The Hamiltonian $H\in\C^\infty$ is assumed to be (partially) unknown due to unstructured uncertainties in the system, which are typically hard to model. The parametric structures of the interconnection matrix $J$, dissipation matrix $R$ and I/O matrix $G$ are assumed to be known, but the parameters themselves might be unknown. Given a dataset of timestamps $\{t_i\}_{i=1}^N$, noisy state observations with inputs, $\{\tilde \x(t_i),\bm{u}(t_i)\}_{i=1}^N$ and a desired trajectory $\x_d(t)\in\R^n$, our aim is to design a tracking control law $\bm{u}_c(\x,\x_d)\in\R^m$ that renders the system~\cref{eq:phs} to a desired PHS given by
\begin{align}\label{for:pchmodel}
        \dot{\bar{\x}}=[J_d(\bar{\x})-R_d(\bar{\x})]\nabla_{\bar{\x}} H_d(\x,\x_d),
\end{align}
where  $\bar{\x}=\x-\x_d$ denotes the tracking error. The dataset $\tilde \x(t_i)$ is assumed to be generated according to $\tilde \x(t_i) = \x(t_i) + \bm{\eta}$ where $\x(t)$ comes from the system~\cref{eq:phs} with zero-mean Gaussian noise $\bm{\eta}\sim\mathcal{N}(\bm{0},\diag[\sigma_1,\ldots,\sigma_n])$. The variances $\sigma_1,\ldots,\sigma_n\in\R_{\geq 0}$ might be unknown.

\section{Tracking Control of PHS}
\label{sec:ctrl}
The problem setting leads to two major challenges: i) rendering the system to a desired PHS without having a comprehensive and accurate model of the system itself, and ii) dealing with a tracking control problem. To overcome these problems, we first use a physics-informed learning approach, namely a Gaussian process port-Hamiltonian system (GP-PHS) model to learn the dynamics of the systems. Due to its Bayesian nature, this model provides not only a prediction of the dynamics of the physical system, but also uncertainty quantification. After we have obtained the model, we use a modified version of the matching equation, which allows us to design a control law that is suitable for the tracking problem and robust against the model uncertainty. In the end, we achieve probabilistic guarantees for the stability of the equilibrium of the tracking error dynamics and prove semi-passivity and boundedness under relaxed assumptions.
\subsection{Gaussian Process Port-Hamiltonian System}\label{sec:GPIntro}
A GP-PHS, introduced in~\cite{9992733}, is a probabilistic model for learning partially unknown PHS based on state measurements. The main idea of GP-PHS is to model the unknown Hamiltonian with a GP while treating the parametric uncertainties in $J,R$ and $G$ as hyperparameters. A Gaussian process $\mathcal{GP}(m_{\mathrm{GP}}(\bm{x}), k(\bm{x},\bm{x}^\prime)$ is a stochastic process on some set $\mathcal{X} \subseteq \R^n$ where any finite collection of points $\bm{x}^1,\ldots,\bm{x}^L\in\X$ follows a multivariate Gaussian distribution
\begin{align*}
    \begin{bmatrix}
        f(\bm{x}^1)\\\vdots\\f(\bm{x}^L)
    \end{bmatrix}
    \!\sim\mathcal{N}\!\left(\!\begin{bmatrix}
        m(\bm{x}^1)\\\vdots\\m(\bm{x}^L)
    \end{bmatrix}\!,\!
    \begin{bmatrix}
        k(\bm{x}^1,\bm{x}^1) & \!\ldots\! & k(\bm{x}^1,\bm{x}^L)\\
        \vdots  & \!\ddots\! & \vdots\\ 
        k(\bm{x}^L,\bm{x}^1) & \!\ldots\! & k(\bm{x}^L,\bm{x}^L)
    \end{bmatrix}\!\right)
\end{align*}
with mean function $m_{\mathrm{GP}}: \R^n \rightarrow \R$, kernel function $k: \R^n \times \R^n \rightarrow \R$, and sample $f \sim \mathcal{GP}(m_{\mathrm{GP}}, k)$. By leveraging that GPs are closed under affine operations, the dynamics of a PHS \cref{eq:phs} is integrated into the GP by
\begin{align}
    \dx&\sim \GP({\hat G}(\x\mid\bm{\varphi}_G)\u,k_{phs}(\x,\x^\prime)),\label{for:gpphs}
\end{align}
where the new kernel function $k_{phs}$ is given by
\begin{align*}
    k_{phs}(\x,\x^\prime)&=\sigma_f^2\hat{J}_R(\x\mid \bm{\varphi}_J,\bm{\varphi}_R)\Pi(\x,\x^\prime)\hat{J}_R^\top(\x^\prime\mid \bm{\varphi}_J,\bm{\varphi}_R)\notag\\
    \Pi_{i,j}(\x,\x^\prime) &= \frac{\partial }{\partial z_i \partial z_j}\exp(-\|\z- \z^\prime\|_{\Lambda}^2)\Big\vert_{\z=\x,\z^\prime=\x^\prime}
\end{align*}
with the Hessian $\Pi\colon\R^n\times\R^n\to\R^{n \times n}$ of the squared exponential kernel, see~\cite{rasmussen2006gaussian}. Thus, the dynamics~\cref{for:gpphs} describes a prior distribution over PHS. The matrices $J,R$ and $G$ of the PHS system~\cref{eq:phs} are estimated by $\hat{J}_R(\x\mid \bm{\varphi}_J,\bm{\varphi}_R)=\hat{J}(\x\mid \bm{\varphi}_J)-\hat{R}(\x\mid \bm{\varphi}_R)$ and $\hat{G}(\bm{x}\mid \bm{\varphi}_G)$. The unknown set of parameters is described by $\bm{\varphi}_J\in\Phi_J\subseteq\R^{n_{\varphi_J}},{n_{\varphi_J}}\in\N$ for the estimated interconnection matrix $\hat{J}(x\vert\bm{\varphi}_J)\in\R^{n\times n} $, $\bm{\varphi}_R\in\Phi_R\subseteq\R^{n_{\varphi_R}},{n_{\varphi_R}}\in\N$ for the estimated dissipation matrix $\hat{R}(x\vert\bm{\varphi}_R)\in\R^{n\times n}$ and $\bm{\varphi}_G\in\Phi_G\subseteq\R^{n_{\varphi_G}},{n_{\varphi_G}}\in\N$ for the estimated I/O matrix $\hat{G}(x\vert\bm{\varphi}_G)\in\R^{n\times m}$. Together with the signal noise $\sigma_f\in\R_{>0}$, the lengthscales $\Lambda=\diag(l_1^2,\ldots,l_n^2)\in\R_{>0}^{n}$ of the kernel $k_{phs}$, the parameter vectors $\bm{\varphi}_J,\bm{\varphi}_R,\bm{\varphi}_G$ are treated as hyperparameters.

We start the training of the GP-PHS by using the collected dataset of timestamps $\{t_i\}_{i=1}^N$ and noisy state observations with inputs $\{\tilde \x(t_i),\bm{u}(t_i)\}_{i=1}^N$ of~\cref{eq:phs} in a filter to create a dataset consisting of pairs of states $ X=[\x(t_1),\ldots,\x(t_N)]\in\R^{n\times N}$ and state derivatives $\dot{X}=[\dx(t_1),\ldots,\dx(t_N)]\in\R^{n\times N}$. 
Then, the unknown (hyper)parameters $\bm\varphi$ can be computed by minimization of the negative log marginal likelihood $-\log \prob(\dot{X}\vert \varphi,X)\sim\dot{X}_0^\top K_{phs}^{-1} \dot{X}_0+\log\vert K_{phs} \vert$, with the mean-adjusted output data  $\dot{X}_0=[[\dx(t_1)-\hat{G}\bm{u}(t_1)]^\top,\ldots,[\dx(t_{N})-\hat{G}\bm{u}(t_{N})]^\top]^\top$. Once the GP model is trained, we can compute the posterior distribution using the joint distribution with mean-adjusted output data $\dot{X}_0$ at a test states $\x^*\in\R^n$. Analogously to the vanilla GP regression, the posterior distribution is then fully defined by the mean $\mu\left(\dx\!\mid\!\x^{*}, \D\right)$ and the variance $\var\left(\dx\!\mid\!\x^{*}, \D\right)$. Similarly, the posterior of the estimated Hamiltonian $\hat{H}(\x^*)=\mu\left(H\!\mid\!\x^{*}, \D\right)$ can be achieved.

At this point, we have achieved a GP-PHS model of the system~\cref{eq:phs}, which provides not only a point estimate but also uncertainty quantification. Assuming that the unknown PHS has a bounded RKHS-norm with respect to the PHS kernel $k_{phs}$, i.e., $\Vert (J-R)\nabla H \Vert_{k_{phs}}< \infty$, the model error can be bounded by
 a constant $p\in(0,1)$ such that $P(|\mu\left(\dot{x}_i\!\mid\!\x, \D\right)-[J(\x)-R(\x)]\nabla H(\x)|\leq \beta_i \var\left(\dot{x}_i\!\mid\!\x, \D\right),\forall i\in\{1,\ldots,n\},x\in\X\})\geq 1-p$. 
The bounded RKHS norm is a reasonable assumption in GP learning, as it limits the class of unknown functions to be learned to the class of functions that the GP model can represent. The value of the constant $\beta$ depends on the number and distribution of the training data. See~\cite{srinivas2012information} for further information on the bounded model error. More detailed information on the training process of GP-PHS can be found in~\cite{9992733}.
\subsection{Passivity-based Tracking control using GP-PHS}
After learning a GP-PHS model of the system to be controlled~\cref{eq:phs}, we will use this model in the next step to design a controller that renders the system passive. As introduced in~\cref{sec:def}, this is often achieved by utilizing interconnection and damping assignment passivity-based control. However, traditional IDA-PBC relies on the matching condition~\cref{eq:matching_equation}, which is primarily suited for shaping potential energy~\cite{ortega2002interconnection}. In contrast, for reference tracking or regulation of certain underactuated systems, modifying kinetic energy is also often necessary. Therefore, we will build on an alternative matching equation, so that IDA-PBC can be extended to reference tracking control based on GP-PHS models. Before we state our main result on a passivity-based tracking controller, we recall that the tracking error is defined by $\bar{\x} = \x - \x_d$, where $\x_d$ is the reference trajectory and introduce the following design property for the desired dynamics~\cref{for:pchmodel}.
\begin{propy}\label{propy:1}
    Let $J_d\colon\X\to\R^n$ be a skew-symmetric matrix, $R_d\colon\X\to\R^n$ a positive semi-definite diagonal matrix, and $H_d\colon\X\times\X\to\R$ the desired Hamiltonian with $\min H_d(\x, \x_d)$ at $\x = \x_d$.
\end{propy}
\begin{thm}\label{thm:1}
Let~\cref{for:gpphs} be a GP-PHS model of the physical system~\cref{eq:phs} based on the dataset $\D$. Given the desired dynamics~\cref{for:pchmodel} with~\cref{propy:1} that satisfy\footnote{in the following, we omit the dependency on $\x,\x_d$ for brevity if obvious.}
\begin{align}
   \!\!\!\hat{G}^\perp\mu\left(\dx\!\mid\!\x, \D\right)\!=\!\hat{G}^\perp \Big([J_d(\bar{\x})\!-\! R_d(\bar{\x})]\nabla_{\bar{\x}} H_d(\x,\x_d)\!+\!\dot{\x}_d\Big)\label{for:spde}
\end{align}
where $H_d,R_d$ are designed such that
\begin{align}
    [\nabla_{\bar{\x}} H_d]^\top \bm{\eta}(\x)\leq[\nabla_{\bar{\x}}H_d]^\top R_d(\bar{\x}) \nabla_{\bar{\x}} H_d\label{for:ineq}
    \end{align}
for all $\{\bm\eta(\x)\in\R^n| |\eta_i(\x)|\leq \beta_i  \var\left(\dot{x}_i\mid\x, \D\right),\,\forall i\in\{1,\ldots,n\},\x\in\X\}$. Then, the control input
\begin{align}
    \bm{u}(\x,\x_d)=&[\hat{G}^\top\hat{G}]^{-1}\hat{G}^\top\big[J_d(\bar{\x})-R_d(\bar{\x})]\nabla_{\bar{x}} H_d(\x,\x_d)\notag\\
    &+\dot{\x}_d-\mu\left(\dx\!\mid\!\x, \D\right)\big]\label{for:ctrl}
\end{align}
for the PHS~\cref{eq:phs} leads to a closed-loop system with a stable equilibrium $\bar{\x}$ on $\X$ with a probability of at least $(1-p)$.
\end{thm}
\begin{proof}
Following~\cite{wang2008modified}, the alternative matching equation for tracking problems is given by $\left[J - R\right] \nabla_\x H = \left[J_d(\bar{\x}) - R_d(\bar{\x})\right] \nabla_{\bar{\x}} H_d (\x, \x_d) + \dot{\x}_d$, where the desired Hamiltonian $H_d(\x, \x_d)$ is designed such that its minimum is where the tracking error vanishes, i.e., $\min H_d(\x, \x_d)$ at $\x = \x_d$. However, we do not have the exact system equations required to compute  the left-hand-side of the alternative matching equation. Instead, the mean prediction of the learned GP-PHS model~\cref{{for:gpphs}} with the probabilistic model error bound is used to rewrite the system~\cref{eq:phs} as 
\begin{align}
    \dx=\mu\left(\dx\!\mid\!\x, \D\right)+\hat G(\x)\bm{u}+\bm{\eta}(\x),\label{prf:uphs}
\end{align}
with perturbation $\bm{\eta}$ as defined in~\cref{thm:1}. Due to the upper-bounded uncertainty shown in~\cref{sec:GPIntro}, there exists a $\bm{\eta}$ with a probability of at least $(1-p)$ such that \cref{eq:phs} equals~\cref{prf:uphs}. Given the satisfaction of the PDE~\cref{for:spde}, the control input~\cref{for:ctrl} applied to~\cref{prf:uphs} leads to
\begin{align}
\dx=&\mu\left(\dx\!\mid\!\x, \D\right)\!+\!\bm{\eta}(\x)\!+\!\hat G\Big(\![\hat{G}^\top\hat{G}]^{-1}\hat{G}^\top[J_d(\bar{\x})\!-\! R_d(\bar{\x})]\notag\\
&\nabla_{\bar{x}} H_d(\x,\x_d)+\dot{\x}_d-\mu\left(\dx\!\mid\!\x, \D\right)\Big)\\
   \Leftrightarrow \dot{\bar{\x}} =&[J_d(\bar{\x})-R_d(\bar{\x})]\nabla_{\bar{x}} H_d(\x,\x_d)+\bm\eta(\x).\label{for:clPCH}
\end{align}
Finally, we choose $H_d$ as a Lyapunov-like function to prove that $\bar{\x}^*=0$ is a stable equilibrium of the closed-loop PHS~\cref{for:clPCH}. The evolution of $H_d$ is given by
\begin{align}
    \dot{H}_d&=[\nabla_{\bar{x}} H_d]^\top [J_d(\bar{\x})-R_d(\bar{\x})] \nabla_{\bar{x}} H_d+[\nabla_{\bar{x}} H_d]^\top \bm{\eta}(\x)\notag\\
    &=-[\nabla_{\bar{x}} H_d]^\top R_d(\bar{\x}) \nabla_{\bar{x}} H_d+[\nabla_{\bar{x}} H_d]^\top \bm{\eta}(\x).\label{for:evolHd}
\end{align}
With~\cref{for:ineq}, that leads to $P(\dot{H}_d\leq 0)\geq 1-p$, which concludes the proof.
\end{proof}
\begin{cor}
The equilibrium $\bar{\x}^*$ will be asymptotically stable with probability $1-p$ if, in addition to~\cref{thm:1}, $\bar{\x}^*$ is an isolated
minimum of $H_d$ and the largest invariant set under the closed-loop dynamics \cref{for:clPCH} contained in
\begin{align}
    \bar{\X}=\{\bar{\x}\in\X\vert[\nabla H_d]^\top R_d(\bar{\x})\nabla H_d=0\}
\end{align}
equals the desired equilibrium $\{\x_d\}$. 
\end{cor}
\begin{proof}
    That is a direct consequence of Proposition 1 in~\cite{ortega2002interconnection}.
\end{proof}
\Cref{thm:1} states that the control law~\cref{for:ctrl} ensures that~$\bar{\x}^*=0$ is a stable equilibrium of the closed-loop system. In addition, the closed-loop dynamics follows the desired PHS~\cref{for:pchmodel} affected by a perturbation~$\bm{\eta}$ that depends on the uncertainty of the GP-PHS model. Thus, a more accurate model, which typically comes with more informative training data, reduces $\bm\eta$ so that the dynamics of the closed-loop~\cref{for:clPCH} converges to the desired dynamics (a perfect model would lead to~$\bm\eta=0$), see~\cite{beckers2023data}. In fact,~\cref{thm:1} ensure that the controller is ``robustified'' against the model error. However, the inequality~\cref{for:ineq} might be challenging to satisfy as the model error~$\bm\eta$ has to converge to zero for~$\x\to\x_d$. In the following, we will relax this condition, leading to boundedness and semi-passivity guarantees. 
\begin{lem}\label{lem:1}
Let~\cref{for:gpphs} be a GP-PHS model of the physical system~\cref{eq:phs} based on the dataset $\D$. Given the desired dynamics~\cref{for:pchmodel} with~\cref{propy:1} that satisfy~\cref{for:spde}, where $H_d$ is radially unbounded and is designed such that
\begin{align}
    [\nabla_{\bar{\x}} H_d]^\top \bm{\eta}(\x)&\leq[\nabla_{\bar{\x}}H_d]^\top R_d(\bar{\x}) \nabla_{\bar{\x}} H_d\notag\\\text{for all } \Vert\x-\x_d \Vert&\geq \epsilon\label{for:ineq2}
    \end{align}
with a constant $\epsilon>0$ for all $\{\bm\eta(\x)\in\R^n| |\eta_i(\x)|\leq \beta_i  \var\left(\dot{x}_i\mid\x, \D\right),\,\forall i\in\{1,\ldots,n\},\x\in\X\}$. Then, the control input~\cref{for:ctrl} for the PHS~\cref{eq:phs} leads to a closed-loop system with a bounded tracking error $\bar{\x}$ on $X$ with a probability of at least $(1-p)$.
\end{lem}
\begin{proof}
    We follow the proof of~\cref{thm:1} and choose~$H_d$ as a Lyapunov-like function, resulting in the evolution of $H_d$ as in~\cref{for:evolHd}. Using~\cref{for:ineq2}, the Lyapunov-like function is decreasing outside a ball around the desired trajectory, i.e., 
\begin{align}
    P(\dot{H}_d\leq 0,\,\forall \Vert\x-\x_d\Vert\geq\epsilon)\geq 1-p,
\end{align}
which proves the boundedness of the tracking error $\bar{\x}$ on $\X$.
\end{proof}
\Cref{lem:1} allows us to have a non-zero model error around the desired trajectory $\x_d$, which, however, leads to the convergence of the tracking error to a \textit{neighborhood} of zero. The size of the neighborhood depends again on the model error and can be explicitly computed by using level-set methods, see~\cite{michel2008stability}. Finally, we prove semi-passivity, which means that the closed-loop system behaves comparable to passive systems outside a ball around the desired trajectory, see~\cite{pogromsky1998passivity}.
\begin{cor}\label{cor:1}
Let $\u_{ex}\in\R^m$ be an external input and the conditions in~\cref{lem:1} be satisfied, than the control input  
\begin{align}
    \bm{u}(\x,\x_d)=&[\hat{G}^\top\hat{G}]^{-1}\hat{G}^\top\big[J_d(\bar{\x})-R_d(\bar{\x})]\nabla_{\bar{x}} H_d(\x,\x_d)\notag\\
    &+\dot{\x}_d-\mu\left(\dx\!\mid\!\x, \D\right)\big]+\u_{ex}\label{for:passcontrol}
\end{align}
renders the system \cref{eq:phs} semi-passive with probability (1-p) with respect to the input $\u_{ex}$ and output $\bm{y}_{ex}=\hat{G}^\top(\bar{\x})\nabla_{\bar{\x}} H_d$.
\end{cor}
\begin{proof}
The evolution of the storage function $H_d$ can be written as
    \begin{align}
    \dot{H}_d&=\![\nabla_{\bar{x}} H_d]^\top\! \big(\![J_d(\bar{\x})\!-\!R_d(\bar{\x})] \nabla_{\bar{x}} H_d+ \bm{\eta}(\x)\!+\!\hat{G}(\bar{\x})\u_{ex}\big)\notag\\
    &=-[\nabla_{\bar{x}} H_d]^\top R_d(\bar{\x}) \nabla_{\bar{x}} H_d+[\nabla_{\bar{x}} H_d]^\top \bm{\eta}(\x)+\bm{y}_{ex}^\top\u_{ex}\notag\\
    &\leq \bm{y}_{ex}^\top\u_{ex} -h(\bar{\x}), 
    \label{for:evolHd1}
\end{align}
with a function $h\colon\X\to\R$. Given~\cref{for:ineq2}, we know that for all $\Vert \x-\x_d\Vert\geq \epsilon$ it follows $h(\bar{\x})>0$ with probability $1-p$. In consequence, the system behaves passively outside a ball around zero, which concludes the proof for semi-passivity, see~\cite{pogromsky1998passivity}.
\end{proof}
In summary, we show that the tracking error dynamics achieves a stable equilibrium at $\bar{\x}^*=0$ if sufficient dissipation is available to counteract the model uncertainties, as required by~\cref{for:ineq}. However, satisfying~\cref{for:ineq} requires that the model error vanishes as the tracking error approaches zero. In~\cref{lem:1,cor:1}, we relax this condition, allowing for a persistent model error, which results in a bounded tracking error and ensures the semi-passivity of the closed-loop dynamics.

\section{Evaluation}
\label{sec:sim}
We consider the problem of designing a tracking control law for an electrostatic microactuator, as shown in~\cref{fig:micro}. The system's dynamic equations in port-Hamiltonian form are given by
\begin{align}\label{for:sim}
        \dx&=\underbrace{\begin{bmatrix}
            0 & 1 & 0\\ -1 & -b & 0\\0 & 0 & -\frac{1}{r}
        \end{bmatrix}}_{J(\x)-R(\x)}\frac{\partial H}{\partial\x}( \x)+\underbrace{\begin{bmatrix}
            0\\0\\\frac{1}{r}
        \end{bmatrix}}_{G(\x)}u\\
        H(\x)&=\frac{1}{2}10(x_1-x_1^s)^2+\frac{1}{2m}x_2^2+\frac{1}{C(x_1)}x_3^2.\notag
\end{align}
with the air gap $x_1$, the momentum $x_2$ and the charge of the device $x_3$, similar to~\cite{maithripala2003nonlinear}. The system is parametrized by the mass of the plate $m=1$ and the capacity $C(x_1)$ that is a function of the distance between the plates. The steady state of the air gap is $x_1^s=1$ and we assume a linear damping of the plate's movement with positive constant $b=0.5$ and a stiffness of $k=10$ for the spring. The input resistance is $r=1$ and $u$ represents the input voltage, which is the control input of the system. We assume that the Hamiltonian $H$ is \textit{unknown to us}, primarily due to the complexity of modeling the function $C$, which is affected by nonlinearities, side effects of the electric field, and other uncertainties. Following the problem formulation in the paper, the objective is to design a tracking controller that ensures the closed-loop system follows a desired trajectory.

To train the GP-PHS model, we first generate a dataset by exciting the microactuator system~\cref{for:sim} with a sinusoidal input signal, given by \( u(t) = \sin(t) \). The system is initialized at \( \x(0) = [0,0,1]^\top \), and data is collected over a time span of \( 0 \si{\milli\second} \) to \( 20 \si{\milli\second} \) with uniform sampling intervals. In total, 300 data pairs \( \{t_i, \x(t_i)\} \) are recorded, Gaussian noise with zero mean and variance \( \sigma^2 = 0.001 \) is introduced to the dataset to account for measurement uncertainties. 

With this dataset, a GP-PHS model is trained according to~\cite[Algorithm 1]{9992733}, which is the basis for the design of a tracking controller following~\cref{thm:1}. We aim to follow a given desired trajectory for the position of the upper plate~$x_1$. The desired trajectory $x_{d,1}$ is assumed to be
\begin{align}\label{eq:desairgap}
    x_{d,1}(t)=x_1^s-0.01t-0.01\sin(0.8t).
\end{align}
\begin{figure}[t]
\begin{center}
\vspace{0.2cm}
	\includegraphics[width=0.8\columnwidth]{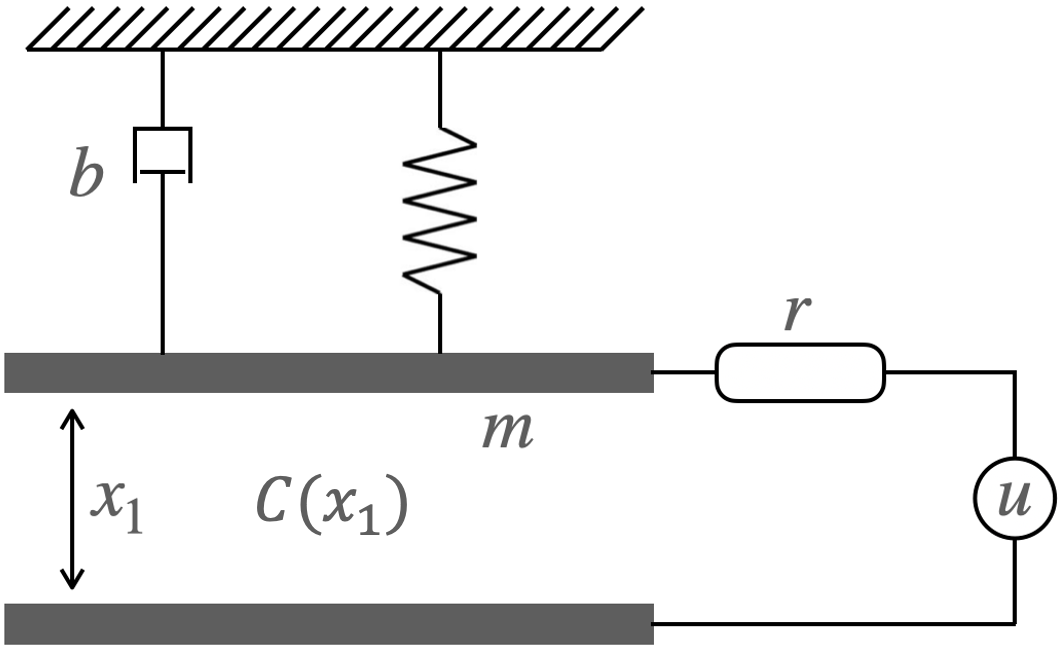}
	\caption{Electrostatic microactuator with unknown capacity $C(x_1)$.}\vspace{-0.6cm}
	\label{fig:micro}
\end{center}
\end{figure}
Next, the controller design requires to solve the matching equation in~\cref{thm:1}. Given the known input matrix $G$, a full-rank left annihilator of $G$ is given by
\begin{align}
    G^\perp(\x)=\begin{bmatrix}
        1 & 0 & 0\\0 & 1 & 0
    \end{bmatrix}\label{eq:anh}
\end{align}
Further, we follow the idea of non-parametric IDA by fixing the desired interconnection and damping matrix to
\begin{align}
    J_d(x)-R_d(x)=\begin{bmatrix}
            0 & 1 & 0\\ -1 & -\hat{b} & 0\\0 & 0 & -\frac{1}{r_d}
        \end{bmatrix},\label{eq:desJR}
\end{align}
where the only change occurs in the damping factor $\frac{1}{r_d}=10$. With~\cref{eq:anh,eq:desJR}, we get the following matching equation~\cref{for:spde}
\begin{align}\label{eq:simmatch}
\nabla_{x_2}\hat{H}(\x)&=\nabla_{\bar{x}_{2}}H_d(\bar{\x})+\dot{x}_{d,1}(t)\\
-\!\nabla_{x_1}\hat{H}(\x)-b\nabla_{x_2}\hat{H}(\x)&=-\!\nabla_{\bar{x}_{1}}H_d(\bar{\x})-b\nabla_{\bar{x}_{2}}H_d(\bar{\x})\notag\\
&\phantom{=}+\dot{x}_{d,2}(t),\notag
\end{align}
 where $\hat{H}(\x)$ is the posterior mean prediction for the Hamiltonian of the trained GP-PHS model, see~\cref{sec:GPIntro}. As candidate for the desired Hamiltonian $H_d$, we use the same posterior mean prediction $\hat{H}(\bar{\x})$ but with the tracking error $\bar{\x}=\x-\x_d$ instead of $\x$. We validate this desired Hamiltonian $\hat{H}(\bar{\x})$ to have its minimum at $\bar{\x}=0$ by point evaluations over a discretized state space, which is set to $\X=[-2,2]^3$. The matching equation~\cref{eq:simmatch} is numerically solved to get 
 $x_{d,2}$ and $x_{d,3}$. Finally, with~\cref{eq:anh,eq:desJR} and $\x_{d}$ in~\cref{for:ctrl}, the control input is computed to
 \begin{align}\label{eq:simctrl}
    u(\x,\x_d)=-\frac{1}{r_d}\nabla_{\bar{x}_{3}}H_d(\bar{\x})+\dot{x}_{d,3}+\frac{1}{r}\nabla_{x_3}\hat{H}(\x).
\end{align}
 
 The top plot of \cref{fig:cl} shows the tracking performance of the closed-loop system with the proposed control law~\cref{eq:simctrl}. As there is uncertainty in the GP-PHS model, the closed-loop trajectory (solid) deviates slightly from the desired air gap $x_{d,1}$ as defined in~\cref{eq:desairgap}. However, according to~\cref{lem:1}, the tracking error remains bounded within a small neighborhood of zero. The evolution of the remaining states, i.e., the momentum and the charge, is visualized in the second plot of \cref{fig:cl}. Finally,~\cref{fig:moredata} supports that the desired Hamiltonian $H_d$ of the tracking error dynamics decreases over time as shown in~\cref{for:evolHd}.
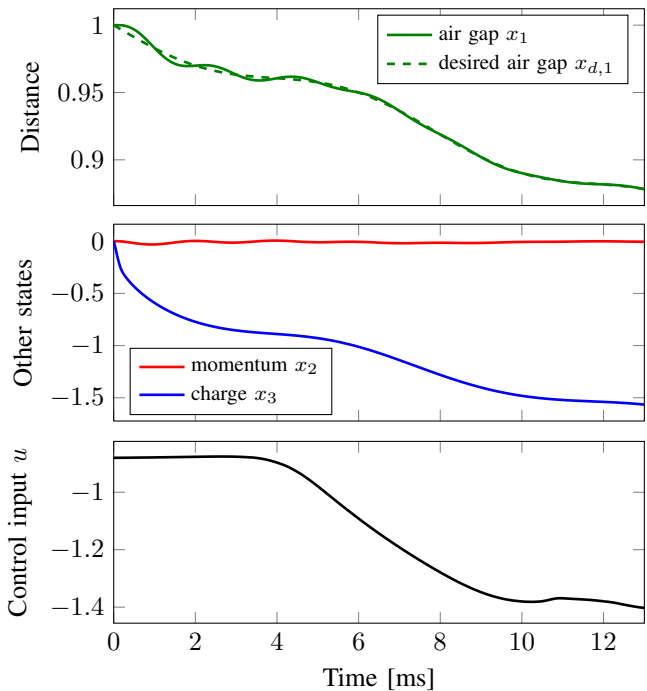
\begin{figure}[t]
\begin{center}
\vspace{0.2cm}
	\input{figure/cl.tex}
    \vspace{-0.6cm}
	\caption{Top 1 and 2: Closed-loop system with the proposed tracking control law. As there is uncertainty in the GP-PHS model, the closed-loop dynamics (solid) slightly deviates from the desired PHS (dashed) but remains in a neighborhood. Bottom: Computed control input over time.}\vspace{-0.6cm}
	\label{fig:cl}
\end{center}
\end{figure}
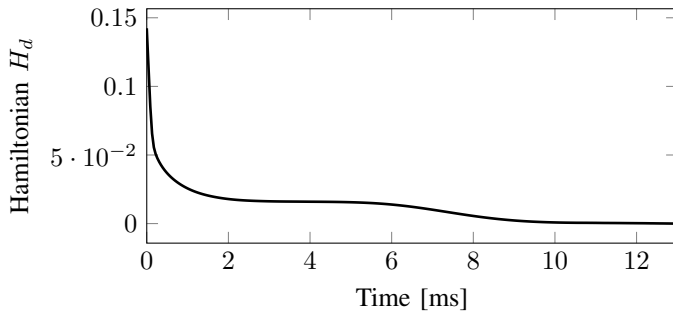
\begin{figure}[ht]
\begin{center}
	\input{figure/moredata.tex}
	\vspace{-0.6cm}\caption{The Hamiltonian function of the tracking error dynamics is decreasing as expected.}\vspace{-0.8cm}
	\label{fig:moredata}
\end{center}
\end{figure}
\section*{Conclusion}
This paper presents a data-driven passivity-based tracking control approach for partially unknown physical systems. The system dynamics are learned using a physics-informed model, specifically a Gaussian process port-Hamiltonian system (GP-PHS). Leveraging the port-Hamiltonian structure, we derive conditions for designing a tracking controller. By incorporating robustness against model errors from the GP-PHS, we ensure the stability of the desired tracking error equilibrium and, under relaxed assumptions, its boundedness and the semi-passivity of the closed-loop system. A simulation study demonstrated the effectiveness of the proposed approach. Future work will focus on comparing this method with other data-driven passivity-based approaches and validating its performance on real-world systems.

\bibliographystyle{IEEEtran}
\bibliography{root}

\end{document}

%% file: mydefs.tex
\newtheorem{thm}{Theorem}
\newtheorem{lem}{Lemma}
\newtheorem{cor}{Corollary}

\newtheorem{propy}{Property}

\newcommand\tran{\mkern-2mu\raise1.25ex\hbox{$\scriptscriptstyle\top\hspace{0.5mm}$}\mkern-3.5mu}
\newcommand{\R}{\mathbb{R}}

\newcommand{\N}{\mathbb{N}}
\newcommand{\C}{\mathcal{C}}

\newcommand{\D}{\mathcal{D}}
\newcommand{\X}{\mathcal{X}}

\newcommand{\bm}[1]{{\boldsymbol{#1}}}

\DeclareMathOperator{\diag}{diag}
\DeclareMathOperator{\var}{var}

\DeclareMathOperator{\prob}{p}
\newcommand{\GP}{\mathcal{GP}}
\newcommand{\z}{\bm z}

\newcommand{\x}{\bm x}

\newcommand{\m}{\bm m}
\newcommand{\dx}{\dot{\bm x}}

\newcommand{\f}{\bm{f}}

\renewcommand{\u}{\bm{u}}

\usepackage[noabbrev]{cleveref} 
\crefname{rem}{Remark}{Remarks}
\crefname{exam}{Example}{Examples}
\crefname{assum}{Assumption}{Assumptions}
\crefname{prop}{Proposition}{Propositions}
\crefname{propy}{Property}{Properties}
\crefname{cor}{Corollary}{Corollaries}
\crefname{lem}{Lemma}{Lemmas}
\crefname{section}{Section}{Sections}
\crefname{thm}{Theorem}{Theorems}
\crefname{alg}{Algorithm}{Algorithms}
\crefname{defn}{Definition}{Definitions}
\crefname{figure}{Fig.}{Fig.}
\Crefname{figure}{Figure}{Figures}
\crefname{equation}{}{}

%% file: figure/cl.tex
\tikzsetnextfilename{cl0}
\begin{tikzpicture}
\begin{axis}[
  name=plot1,
  ylabel={Distance},
  legend pos=north west,
  width=\columnwidth,
  height=4.2cm,
  xmin=0,
  xmax=13,
  xticklabels={},
  legend style={font=\footnotesize},
  legend cell align={left},
  legend pos=north east]
\addplot[color=green!50!black,line width=1pt,no marks] table [x index=0,y index=2]{data/control.dat};
\addplot[color=green!50!black,dashed,line width=1pt,no marks] table [x index=0,y index=1]{data/control.dat};
\legend{air gap $x_1$,desired air gap $x_{d,1}$};
\end{axis}
\begin{axis}[
  name=plot2,
   at=(plot1.below south east), anchor=above north east,
  ylabel={Other states},
  legend pos=north west,
  width=\columnwidth,
  height=4.2cm,
  xmin=0,
  xmax=13,
  xticklabels={},
  legend style={font=\footnotesize},
  legend cell align={left},
  legend pos=south west]
\addplot[color=red,line width=1pt,no marks] table [x index=0,y index=3]{data/control.dat};
\addplot[color=blue,line width=1pt,no marks] table [x index=0,y index=4]{data/control.dat};
\legend{momentum $x_2$, charge $x_3$};
\end{axis}
\begin{axis}[
  name=plot3,
   at=(plot2.below south east), anchor=above north east,
  xlabel={Time [ms]},
  ylabel={Control input $u$},
  legend pos=north west,
  width=\columnwidth,
  height=4cm,
  xmin=0,
  xmax=13,
  legend style={font=\footnotesize},
  legend cell align={left},
  legend pos=north east]
\addplot[color=black,line width=1pt,no marks] table [x index=0,y index=1]{data/control_u.dat};
\end{axis}
\end{tikzpicture} 

%% file: figure/moredata.tex
\tikzsetnextfilename{error}
\begin{tikzpicture}
\begin{axis}[
  name=plot1,
  xlabel={Time [ms]},
  ylabel={Hamiltonian $H_d$},
  legend pos=north west,
  width=\columnwidth,
  height=4.7cm,
  xmin=0,
  xmax=13,
  legend style={font=\footnotesize},
  legend cell align={left},
  legend pos=north east]
\addplot[color=black,line width=1pt,no marks] table [x index=0,y index=1]{data/Lyap.dat};
\end{axis}
\end{tikzpicture} 